\documentclass[conference]{IEEEtran}
\IEEEoverridecommandlockouts
\usepackage{cite}
\usepackage{amsmath,amssymb,amsfonts}
\usepackage{algorithmic}
\usepackage{graphicx}
\usepackage{comment}
\usepackage{amsthm}
\usepackage{textcomp}
\usepackage{xcolor}
\usepackage[ruled,vlined]{algorithm2e}
\usepackage{babel}
\usepackage{lipsum}
\usepackage{float}
\floatstyle{plaintop}
\restylefloat{table}
\usepackage[tableposition=top]{caption}
\newtheorem{lma}{Lemma}
\newtheorem{thrm}{Theorem}

\begin{document}

\title{Online-Learning Deep Neuro-Adaptive Dynamic Inversion Controller for Model Free Control\\
\thanks{This paper was funded by the NASA Missouri Space Grant Consortium. We appreciate all of their generous aid.

This paper is dedicated to my late advisor, Dr. S. N. Balakrishnan, who unfortunately passed before its submission. His helpful aid was instrumental in the development of this document.}
}

\author{
\IEEEauthorblockN{Nathan Lutes\IEEEauthorrefmark{1}, K. Krishnamurthy\IEEEauthorrefmark{1}, Venkata Sriram Siddhardh Nadendla\IEEEauthorrefmark{2}, S. N. Balakrishnan\IEEEauthorrefmark{1}}
\\[-0.75ex]
\IEEEauthorblockA{\IEEEauthorrefmark{1}Mechanical and Aerospace Dept., \IEEEauthorrefmark{2}Computer Science Dept. \\
Missouri University of Science and Technology\\
Rolla, MO, USA\\
Email: \{nalmrb, kkrishna, nadendla, bala\}@mst.edu}
% \and
% \and
}

\maketitle

\begin{abstract}
Adaptive methods are popular within the control literature due to the flexibility and forgiveness they offer in the area of modelling. Neural network adaptive control is favorable specifically for the powerful nature of the machine learning algorithm to approximate unknown functions and for the ability to relax certain constraints within traditional adaptive control. Deep neural networks are large framework networks with vastly superior approximation characteristics than their shallow counterparts. However, implementing a deep neural network can be difficult due to size specific complications such as vanishing/exploding gradients in training. In this paper, a neuro-adaptive controller is implemented featuring a deep neural network  trained on a new weight update law that escapes the vanishing/exploding gradient problem by only incorporating the sign of the gradient. The type of controller designed is an adaptive dynamic inversion controller utilizing a modified state observer in a secondary estimation loop to train the network. The deep neural network learns the entire plant model on-line, creating a controller that is completely model free. The controller design is tested in simulation on a 2 link planar robot arm. The controller is able to learn the nonlinear plant quickly and displays good performance in the tracking control problem.
\end{abstract}


\section{Introduction}
Adaptive control methods have been on the rise in recent years within the control theory literature. Their ability to correct errors in modelling or to capture and cancel unknown disturbances have proven to be quite desirable as more 'intelligent' and resilient control methods are sought. A quite popular and very powerful form of adaptive control has been found in controllers featuring artificial neural networks (ANNs). ANNs are machine learning algorithms with the extremely useful ability to model, or 'learn', most any function within some degree of error and have been used extensively for other tasks such as learning large and complex data sets in the realm of data science. Within the scope of control theory, ANNs are primarily utilized as estimation methods within controllers to learn uncertainty in plant dynamics or disturbances online and thus improve the controller performance over time. For example, \cite{rajagopal_balakrishnan_nguyen_krishnakumar_mannava_2009} and \cite{padhi_balakrishnan_unnikrishnan_2007} created adaptive dynamic inversion controllers by using a neural network state observer, termed modified state observer (MSO), in a separate estimation loop to learn modelling uncertainty or unknown disturbances in non-affine, non-square nonlinear systems. In \cite{ghafoor_balakrishnan_jagannathan_yucelen_2018}, this type of control scheme was extended to an event triggered control framework. Other examples of neural network control are prevalent throughout the literature for many different applications from aerospace to autonomous vehicles to robotics.
% \begin{comment}
% or to provide real-time solutions for complex computation problems central to the control formulation.
% \end{comment}
In \cite{6891229}, a missile guidance law was created by transforming a model predictive control scheme into a constrained quadratic programming problem and solved online using neural networks in a finite receding horizon. In \cite{8559329}, a tracking control solution for unmanned surface vehicle's was created based on a combination of Deep Deterministic Policy Gradient (DDPG) and neural network control. In \cite{9097881}, a neural network was employed to approximate plant dynamics and external disturbances for use in combining active disturbance attenuation (ADA) with robust dynamic programming (RADP). Recurrent neural networks were used in \cite{8489118} in the form of long short term memory neural networks (LSTM-NN) for an attitude estimation application to unmanned aerial vehicles (UAV)'s by training the network on attitude time-series data. 
% \begin{comment}
% or used offline to give approximate solutions to complex functions, such as the Hamilton-Jacobi-Bellman function in optimal control.
% \end{comment}

Recent advances in machine learning have lead to the development of deep neural networks (DNNs) which are even more powerful, and duely more complex, versions of ANNs, usually with considerably larger structures. In \cite{7867471}, an in-depth look at DNNs is provided, discussing their development, popular models, applications and future research directions. Within the realm of control theory, DNNs are commonly used for reinforcement learning. The deep reinforcement learning (DRL) framework is presented in \cite{8203866} and then its applications and some current implementations are discussed.
% \begin{comment}
% The DRL framework largely relies on an initial offline training session. The following papers present DNN's in a more control specific manner.
% \end{comment}
A DRL based neural combinatorial optimization strategy for online vehicle route generation is developed in \cite{8693516}. Liang \emph{et al} \cite{8682051} created a novel guidance scheme based on model-based DRL where a DNN is trained as a guidance dynamics predictive model and incorporated into a model predictive path integral control framework. Xin \emph{et al} introduces a hybrid learning strategy for DRL control of commercial aircraft where a PID baseline controller is utilized during the preliminary online training of the DNN's and then the DNN calculated control policy is utilized once it produces better performance than the baseline controller in \cite{8757756}. In \cite{7487175}, model predictive control is used to train a DRL controller for control of a quadcopter drone. Sarabakha \emph{et al} \cite{8794314} uses DRL for UAV trajectory tracking by employing a baseline controller that simultaneously performs initial control and trains the DNN-based controller. Once the DNN controller produces sufficient policies, control switches to the neural network controller however online training continues for different sets of trajectories not used in the training phase.

DNN's are sometimes used in other applications besides DRL. For example, \cite{9123670} used DNN's  as approximate controllers by training on data created by model predictive control laws. Sarabakha \emph{et al} \cite{8809217} combined deep learning and fuzzy logic to create an online learning method for improved control of nonlinear systems. The DNN was pretrained offline on the recorded input-output dataset while a conventional controller controlled the system. Then when the DNN controller produced better results than the conventional, control was switched, with the neural network controller continuing to be trained online. R.  Chai  \emph{et  al} \cite{8939337} used DNNs for real-time attitude control on a six-DOF hypersonic vehicle reentry flight application. The neural networks were trained offline using pre-generated optimal trajectory data then implemented in real-time.
Many of the previous literature utilize some form of offline pre-training or otherwise initially use baseline models whilst the neural network learns the system and uncertainties and/or control policies. This is because unlike shallow neural networks, DNN's are much more difficult to train and typically take substantially more time. They are also more susceptible to problems such as vanishing/exploding gradient, making them difficult to use in purely online roles. However, with a recent advance in the training of neural networks, the power of DNN's can now be harnessed for online learning.

Recently, Bernstein \emph{et al} \cite{bernstein2021learning} discovered a new learning algorithm where the weights of a neural network can be updated individually via adding or subtracting a proportion of the weight in the direction that reduces the loss function gradient. In practice, this is approximated by multiplying by the exponential function with the parameter being a learning rate factor. The attribute of this algorithm that makes it especially attractive for training neural networks is the fact that the value of the gradient is not used, only its sign. This allows the algorithm to escape the vanishing/exploding gradient problem that has previously made it very difficult to train neural networks of a significant size in an online fashion. The discovery of the algorithm was found after Bernstein \emph{et al} defined a new distance between neural networks that could be used to facilitate training of one neural network to a theoretical 'optimal' neural network in \cite{bernstein2021distance}. Termed \emph{deep latent trust} \cite{bernstein2021distance}, this new distance definition allows for the relative change in loss function gradient as the result of network weight perturbations to be expressed as a product of the relative perturbations themselves. This in turn allows for definitions of different weight update laws that can be guaranteed to descend the loss function gradient and thus converge to the optimal weights.

This paper presents a novel control framework featuring a dynamic inversion controller combined with the Deep Neural Network Modified State Observer (DNN-MSO) - a model-less observer containing a DNN that trains the neural network online and provides a complete estimate of the plant dynamics. The observer converges to the true system by training the DNN on the error between the estimated system and the true system measurements using the multiplicative weight update tuning law and also uses an estimation error feedback term to improve estimation convergence properties. This smooths the control input and helps to mitigate high frequency chatter in the controlled system while the neural network is converging to the system dynamics. The neural network estimation property combined with the convergence guarantees in the weight tuning law proves that the DNN can adequately learn the system dynamics. This estimation is then used to replace the dynamics with the desired dynamics in the dynamic inversion controller. A Lyapunov analysis reinforces that the closed loop system is uniformly ultimately bounded with the bound being dependent on the controller feedback gain matrix and the estimation capabilities of the neural network. The controller is tested in simulation on a highly nonlinear robotic manipulator.

The body of the paper can be summarized as follows: in section 3, the DNN-MSO is explained and the controller formulation is given ending with the derivation of the tracking error dynamics; in section 4, the multiplicative weight update law and its derivation are presented and the performance guarantees for certain learning rate values are established; in section 5, the stability analysis of the error dynamics is explored using Lyapunov analysis and the concluding results are discussed; in section 6 the controller is tested in simulation via a tracking control problem of a  robot manipulator and the results are discussed and finally the paper concludes in section 7. 

\section{Controller Formulation}
% - describe deep neural network
% - talk about MSO formulation
% - dynamic inversion formulation
% - system error dynamics (?)

Consider the dynamics of an uncertain, nonlinear system:
\begin{equation}
    \dot{X}(t) = f(X(t)) + B(u(t)+d(X(t)))
\end{equation}
where $f(X(t))$ is the nonlinear, known system dynamics, B is the control matrix, $u(t)$ is the control signal and $d(X(t))$ is the state dependent uncertainty. The system can be simplified by distributing $B$:
\begin{equation}
    \dot{X}(t) = f(X(t)) + \check{f}(X(t)) + Bu(t)
\end{equation}
\begin{equation*}
    \check{f}(X(t)) = Bd(X(t))
\end{equation*}
The performance goal is to drive the uncertain system towards a desired system having dynamics described by:
\begin{equation}
    \dot{X}_{d}(t) = f^{*}(X_{d}(t),u_{n})
\end{equation}
where $u_{n}$ is some nominal control signal. To accomplish this, we employ a dynamic inversion controller to replace the original system's dynamics with the desired dynamics. However, a significant problem with dynamic inversion controllers is their sensitivity to modelling errors in the system dynamics and thus the uncertainty in the plant poses a significant challenge. To remedy this, past literature \cite{rajagopal_balakrishnan_nguyen_krishnakumar_mannava_2009}, \cite{padhi_balakrishnan_unnikrishnan_2007}, \cite{ghafoor_balakrishnan_jagannathan_yucelen_2018} have relied on learning the uncertainty using a modified state observer (MSO).

The MSO is a neural network observer that allows estimation of the system uncertainty in a secondary estimation loop. The typical MSO uses the known dynamics as a baseline and builds upon this using a neural network that learns the uncertainty online by comparing the state estimation and the true state measurement. Finally, the MSO contains an estimation error feedback term which helps in diminishing any high frequency dynamics that might be introduced into the controller. The equation for the MSO is given by:
\begin{equation}
    \dot{\hat{X}}(t) = f(\hat{X}(t))+B(u(t)+\hat{d}(X(t))-K_{2}(X(t)-\hat{X}(t))
\end{equation}
Where $f(\hat{X}(t))$ is the known dynamics baseline, $\hat{d}(X(t))$ is the uncertainty estimate and $K_{2}(X(t)-\hat{X}(t))$ is the feedback term which helps to filter the control signal. Using the new learning algorithm provided by \cite{bernstein2021learning}, we can expand upon the MSO concept.

The DNN-MSO is a special modification of the MSO featuring a deep neural network that is completely data driven. The DNN learns the complete plant dynamics and any disturbances online. The modification of the original MSO to the DNN-MSO is quite simple with the DNN-MSO being described by:
\begin{equation}\label{DNN-MSO}
    \dot{\hat{X}}(t) = \hat{f}(X(t))+Bu(t) -K_{2}(X(t)-\hat{X}(t))
\end{equation}
One can see from (\ref{DNN-MSO}), that the known dynamics and uncertainty estimation has been replaced with just one estimation term, $\hat{f}(X(t))$, that represents the entire dynamic model of the plant. 

$\hat{f}(X(t))$ is the output of the deep neural network, in this case a deep multi-layered perceptron (MLP), which can be described mathematically as:
\begin{multline}
    \hat{f}(X(t))=\\
    W_{n-1}\sigma_{n-1}(W_{n-2}\sigma_{n-2}(W_{n-3}...W_{1}\sigma_{1}(W_{in}X(t))))
\end{multline}
where $W_{i}$ is the weights connecting layer $i$ to layer $i+1$ ($W_{in}$ connects the input of the network to layer 1), $\sigma_{i}$ is the output of layer {i} and $X(t)$ is the measured states. Note that for mathematical simplicity and convenience, the biases of the network are assumed implicit and thus omitted from the notation. Further note that $\sigma_{n}$ has been omitted because the activation function of the final layer is linear as is standard procedure for a MLP network with a continuous output. Now that the new estimation model has been described, the controller can be formulated.

Since the desired objective of the controller is to drive the true system $X(t)$ to the desired system $X_{d}(t)$, consider the following equation:
\begin{equation}\label{contForm1}
    \dot{X} - \dot{X}_{d} + K(X(t)-X_{d}(t)) = 0
\end{equation}
which enforces first order asymptotically stable error dynamics \cite{padhi_balakrishnan_unnikrishnan_2007}. Note that $K$ is a positive-definite, user-defined gain matrix. Substituting the system dynamics into (\ref{contForm1}):
\begin{multline}
     f(X(t)) + \check{f}(X(t)) + Bu(t) - f^{*}(X_{d}(t),u_{n})\\ + K(X(t)-X_{d}(t))=0
\end{multline}
Rearranging:
\begin{multline}
    Bu(t) = f^{*}(X_{d}(t),u_{n})-(f(X(t)) + \check{f}(X(t)))\\-K(X(t)-X_{d}(t))
\end{multline}
Assuming that the uncertainty estimate, $\hat{f}(X(t))$, is available and $\hat{f}(X(t)) \approx{f(X(t)) + \check{f}(X(t))}$, this can be substituted for the system dynamics as:
\begin{equation}
    Bu(t) = f^{*}(X_{d}(t),u_{n})-\hat{f}(X(t))-K(X(t)-X_{d}(t))
\end{equation}
Then, if $B$ is invertible, the dynamic inversion controller can be completed by multiplying by the inverse as:
\begin{multline}
    u(t) =B^{-1}( f^{*}(X_{d}(t),u_{n})-\hat{f}(X(t))\\-K(X(t)-X_{d}(t)))
\end{multline}
If $B$ is not invertible, than slack variables can be added in the manner discussed in \cite{rajagopal_balakrishnan_nguyen_krishnakumar_mannava_2009}. This amounts to adding and subtracting the slack matrix $B_{s}$ and the slack control signal $u_{s}$ to the formulation to augment the control matrix into an invertible form. The following derivation illustrates this point.
\begin{multline*}
    Bu(t) = f^{*}(X_{d}(t),u_{n})-\hat{f}(X(t))\\-K(X(t)-X_{d}(t)) + B_{s}u_{s} - B_{s}u_{s}
\end{multline*}
\begin{multline*}
    Bu(t)+ B_{s}u_{s} = f^{*}(X_{d}(t),u_{n})-\hat{f}(X(t))\\-K(X(t)-X_{d}(t)) + B_{s}u_{s} 
\end{multline*}
\begin{equation*}
    \bar{B} = [B\ B_{s}], \quad 
    \bar{U} = [u\ u_{s}]^{T}
\end{equation*}
\begin{multline*}
    \bar{B}\bar{U}=f^{*}(X_{d}(t),u_{n})-\hat{f}(X(t))\\-K(X(t)-X_{d}(t)) + B_{s}u_{s}
\end{multline*}
\begin{multline*}
    \bar{U}=\bar{B}^{-1}(f^{*}(X_{d}(t),u_{n})-\hat{f}(X(t))\\-K(X(t)-X_{d}(t)) + B_{s}u_{s})
\end{multline*}
Note that the actual control will need to be extracted from $\bar{U}$ and that the choice of slack variables affects the control signal. A methodology for selection of slack variable values is discussed in \cite{padhi_balakrishnan_unnikrishnan_2007}. Now that the controller has been defined, the error dynamics of the closed loop system can be discussed.

First, the tracking error, estimation error and estimate tracking error are defined as:
\begin{align}
    e_{r} &\triangleq{X(t) - X_{d}(t)}\\ e_{a} &\triangleq{X(t)-\hat{X}(t)}\\ \hat{e}_{r}&\triangleq{\hat{X}(t)-X_{d}(t)}
\end{align}
Now the tracking error can be rewritten as a combination of the estimation error and the estimate tracking error:
\begin{equation*}
    e_{r} = X(t) - \hat{X}(t)+\hat{X}(t) -X_{d}(t)
\end{equation*}
\begin{equation}\label{errDy1}
    e_{r}=e_{a}+\hat{e}_{r}
\end{equation}
Taking the derivative of (\ref{errDy1}):
\begin{equation}\label{erdot}
    \dot{e_{r}}=\dot{e_{a}}+\dot{\hat{e}}_{r}
\end{equation}
Expanding the estimation error derivative:
\begin{equation}\label{eadotfinal}
    \begin{array}{lcl}
    \dot{e}_{a} &=& \dot{X}(t) - \dot{\hat{X}}(t)  \\[2ex]
     &=& f(X(t)) + \check{f}(X(t)) + Bu(t) -\\ && \qquad \qquad (\hat{f}(X(t))+Bu -K_{2}(X(t)-\hat{X}(t)))
     \\[2ex]
    &=& f(X(t)) + \check{f}(X(t)) - \hat{f}(X(t))\\
    && \qquad \qquad
    +K_{2}(X(t)-\hat{X}(t))
    \end{array}
\end{equation}
% \begin{equation*}
%     \dot{e}_{a} = \dot{X}(t) - \dot{\hat{X}}(t)
% \end{equation*}
% \begin{multline*}
%     \dot{e}_{a}= f(X(t)) + \check{f}(X(t)) + Bu(t)\\ - (\hat{f}(X(t))+Bu -K_{2}(X(t)-\hat{X}(t)))
% \end{multline*}
% \begin{equation}\label{eadotfinal}
%     \dot{e}_{a}= f(X(t)) + \check{f}(X(t)) - \hat{f}(X(t)) +K_{2}(X(t)-\hat{X}(t))
% \end{equation}
Expanding the estimate tracking error derivative:
\begin{equation*}
    \begin{array}{lcl}
    \dot{\hat{e}}_{r}&=&\dot{\hat{X}}(t)-\dot{X}_{d}(t)
    \\[2ex]
    &=& \hat{f}(X(t))+Bu(t) -K_{2}(X(t)-\hat{X}(t))\\
    && \qquad \qquad -f^{*}(X_{d}(t),u_{n})
    \end{array}
\end{equation*}
% \begin{equation*}
%     \dot{\hat{e}}_{r}=\dot{\hat{X}}(t)-\dot{X}_{d}(t)
% \end{equation*}
% \begin{multline*}
%     \dot{\hat{e}}_{r}=\hat{f}(X(t))+Bu(t) -K_{2}(X(t)\\-\hat{X}(t))-f^{*}(X_{d}(t),u_{n})
% \end{multline*}
Now using $Bu(t) = f^{*}(X_{d}(t),u_{n})-\hat{f}(X(t))-K(X(t)-X_{d}(t))$:
\begin{equation}\label{erhatdotfinal}
    \begin{array}{lcl}
    \dot{\hat{e}}_{r}&=&\hat{f}(X(t))+f^{*}(X_{d}(t),u_{n})-\hat{f}(X(t))\\
    && \qquad \qquad  -K(X(t)-X_{d}(t))-K_{2}(X(t)\\
    && \qquad \qquad -\hat{X}(t))-f^{*}(X_{d}(t),u_{n})
    \\[2ex]
    &=& -(K(X(t)-X_{d}(t))+K_{2}(X(t)-\hat{X}(t)))
    \end{array}
\end{equation}
% \begin{multline*}
%     \dot{\hat{e}}_{r}=\hat{f}(X(t))+f^{*}(X_{d}(t),u_{n})-\hat{f}(X(t))-K(X(t)-X_{d}(t)) \\-K_{2}(X(t)-\hat{X}(t))-f^{*}(X_{d}(t),u_{n})
% \end{multline*}
% which simplifies to:
% \begin{equation}\label{erhatdotfinal}
%     \dot{\hat{e}}_{r}=-(K(X(t)-X_{d}(t))+K_{2}(X(t)-\hat{X}(t)))
% \end{equation}
Finally, inputting (\ref{eadotfinal}) and (\ref{erhatdotfinal}) into (\ref{erdot}):
\begin{multline*}
    \dot{e}_{r} = f(X(t)) + \check{f}(X(t)) - \hat{f}(X(t)) +K_{2}(X(t)-\hat{X}(t)) \\+ (-(K(X(t)-X_{d}(t))+K_{2}(X(t)-\hat{X}(t))))
\end{multline*}
which simplifies to:
\begin{equation}\label{erdotfinal}
    \dot{e}_{r}=\tilde{f}(X(t)) - K(X(t)-X_{d}(t))
\end{equation}
\begin{equation*}
    \tilde{f}(X(t)) = f(X(t)) + \check{f}(X(t)) - \hat{f}(X(t))
\end{equation*}

\section{Multiplicative Weight Update}
% - talk about how algorithm is implemented in controller
% - talk about the descent theorem and the corresponding math, reference papers

The backbone of any online neuro-adaptive controller is its ability to learn quickly and robustly.
In this case, the driving force behind the DNN-MSO is the novel multiplicative weight update law \cite{bernstein2021learning}.
As its name suggests, this update law is based on the multiplicative weight update algorithm discussed in \cite{DBLP:journals/toc/AroraHK12}. The equation for this tuning law is given simply as:
\begin{equation}\label{MUL}
    W_{i,l}  = W_{i,l}e^{\pm \eta}
\end{equation}
More precisely, weight $i$ of layer $l$ is updated by multiplying by an exponential term raised to the power of the learning rate $\eta$. In this manner, each weight is updated individually with the sign of $\eta$ to be determined later. The local updating property of the algorithm allows it to be invariant to complex architectures and vaguely defined layers. For the purpose of online neuro-adaptive control, the fundamental advantage of this type of update law is that it does not depend on gradient information and thus escapes the exploding/vanishing gradient problem making online training of deep neural networks feasible. The exact algorithm, the multiplicative adaptive moments based optimizer, provided by \cite{bernstein2021learning} is given below:

\begin{algorithm}
\SetAlgoLined
Select Algorithm Hyperparameters: $\eta, \eta^*, \sigma^{*}, \beta$
\\
Initialize Weights
\\
$\bar{g} \gets 0$ \qquad Initialise second moment estimate
\\
\Repeat{converged}{
$g \gets \ \text{StochasticGradient}()$
\\
$\bar{g}^{2} \gets (1-\beta) g^{2} + \beta \bar{g}^{2}$
\\
$W \gets \ W \odot \exp[-\eta \cdot \text{sign}(W) \odot \text{clamp}_{\eta^{*}/\eta}(g/\bar{g})]$
\\
$W \gets \ \text{clamp}_{\sigma^{*}}(W)$
}
\caption{Multiplicative Adaptive Moments}
\end{algorithm}

The hyperparameters of the algorithm are: the learning rate $(\eta)$, the maximum weight perturbation factor $(\eta^{*})$, the maximum weight value $(\sigma^{*})$ and the convex combination factor $(\beta)$. Note that the algorithm has a hard restriction on the size of the weights, set by the user, which enforces a hard upper and lower bound on the network weights. Although a learning rate is specified, the actual weight perturbation factor magnitude may be more, up to a user-defined maximum, so the maximum weight change for one update is hard-bounded as well. The actual update mechanism is very close to the vanilla multiplicative update rule shown in (\ref{MUL}) with a few changes for increased practicality. For example, $g/\bar{g}$ is $O(1)$ and therefore is just a higher precision version of $sign(g)$ which captures more information about the gradient without subjecting the algorithm to the exploding/vanishing gradient problems. It is also obvious that the sign of the update factor is determined based off of $sign(W)$ and $sign(g)$ with the reasoning behind this discussed in the following lemma and theorem.

The theory behind the multiplicative weight update tuning law is fundamentally based on a novel definition of functional distance for neural networks termed \textit{deep relative trust}, introduced in \cite{bernstein2021distance}, which is described as follows. Consider a neural network having $L$ layers and weight parameters $W = (W_{1}, W_{2},...,W_{L})$ with weight perturbations $\Delta W = (\Delta W_{1}, \Delta W_{2},...,\Delta W_{L})$. Consider a loss function over the network parameters: $\mathcal{L}(W)$. Let $g_{k}(W) = \nabla_{w_{k}}\mathcal{L}(W)$ denote the gradient of the loss. Then the gradient breakdown is bounded by:
\begin{align}\label{DRT}
    \frac{\|g_{k}(W+\Delta W)-g_{k}(W)\|_{F}}{\|g_{k}(W)\|_{F}} &\leq{\prod_{l=1}^{L} \bigg(1+\frac{\|\Delta W_{l}\|_{F}}{\|W_{l}\|_{F}}\bigg)-1} \nonumber\\
    % (for\ all\ k &= 1,...,L)
\end{align}
for all $k=1,...,L$. As can be seen from (\ref{DRT}), the relative change in Loss function gradient with respect to each weight layer is upper bounded by a product of the relative perturbation of the weights over all layers. Thus, we have a tractable model for gradient breakdown with which we can build a descent lemma which is the cornerstone of the Multiplicative Adaptive Moments optimiser theorem.

Next, we must introduce a general lemma for guaranteeing gradient descent.
\begin{lma}\label{descentLem}
Consider a continuously differentiable function $\mathcal{L}: \mathbb{R}^{n}\mapsto \mathbb{R}$ that maps $W \mapsto \mathcal{L}(W)$. Define $W \triangleq (W_{1}, W_{2},...,W_{L})$ where $L$ is the total number of parameter groups. Similarly, define $\Delta W \triangleq (\Delta W_{1}, \Delta W_{2},...,\Delta W_{L})$ as the perturbation of $W$. Let $\theta_{k}$ measure the angle between $\Delta W_{k}$ and the negative gradient $-g_{k}(W) = -\nabla_{W_{k}}\mathcal{L}(W)$. Then:
\begin{multline}
\mathcal{L}(W+\Delta W)-\mathcal{L}(W) \leq -\sum_{k=1}^{L} \|g_{k}(W)\|_{F}\|\Delta W_{k}\|_{F}\\\Bigg[cos(\theta_{k})-\max_{t\in[0,1]} \frac{\|g_{k}(W+t\Delta W)-g_{k}(W)\|_{F}}{\|g_{k}(W)\|_{F}}\Bigg]
\end{multline}
\end{lma}
\begin{proof}
By the fundamental theorem of Calculus, we have:
\begin{multline}\label{lmaPeq1}
    \mathcal{L}(W+\Delta W)-\mathcal{L}(W) = \sum_{k=1}^{L}\Bigg[g_{k}(W)^{T}\Delta W_{k}+\\\int_{0}^{1} [g_{k}(W+t\Delta W)-g_{k}(W)]^{T}\Delta W_{k} \,dt \Bigg]
\end{multline}
Replacing the first term on the right side of (\ref{lmaPeq1}) with the cosine formula for the dot product and using the integral estimation lemma to bound the second term, we arrive at:
\begin{multline}
\mathcal{L}(W+\Delta W)-\mathcal{L}(W) \leq -\sum_{k=1}^{L} \|g_{k}(W)\|_{F}\|\Delta W_{k}\|_{F}\\\Bigg[cos\theta_{k}-\max_{t\in[0,1]} \frac{\|g_{k}(W+t\Delta W)-g_{k}(W)\|_{F}}{\|g_{k}(W)\|_{F}}\Bigg].
\end{multline}
\end{proof}
Lemma \ref{descentLem} now gives us our formal descent guarantee conditions. More precisely, this is when:
\begin{align}
    \max_{t\in[0,1]} \frac{\|g_{k}(W+ t\Delta W)-g_{k}(W)\|_{F}}{\|g_{k}(W)\|_{F}}<cos(\theta_{k})\nonumber \\(for\ k = 1,...,L)
\end{align}
We see that the descent guarantee stems from the gradient breakdown. Now we can use our previously defined notion of \textit{deep relative trust} to derive the condition over the learning rate in which descent is guaranteed. This is expressed in the following theorem.
\begin{thrm}
Let $\mathcal{L}$ by the continuously differentiable loss function of a neural network of depth $L$ that obeys deep relative trust. For $k = 1,...,L$, let $0\leq{\gamma_{k}}\leq{\frac{\pi}{2}}$ denote the angle between $|g_{k}(W)|$ and $|W_{k}|$ (where $|\centerdot|$ denotes element-wise absolute value). Then the following update law:
\begin{equation}\label{UpdL}
    W \xrightarrow{} W + \Delta W = W \odot [1-\eta\ sign\ W \odot sign(g(W))]
\end{equation}
guarantees a decrease in the loss function provided that:
\begin{equation}
    \eta < (1+cos(\gamma_{k})^{\frac{1}{L}}-1, \quad (for\ all\ k = 1,...,L)
\end{equation}
\end{thrm}
\begin{proof}
    Using the definition of \textit{deep relative trust} (\ref{DRT}), we have:
    \begin{multline}
        \max_{t\in[0,1]} \frac{\|g_{k}(W+t\Delta W)-g_{k}(W)\|_{F}}{\|g_{k}(W)\|_{F}} \\\leq{\max_{t\in[0,1]}\prod_{l=1}^{L} \bigg(1+\frac{\|t\Delta W_{l}\|_{F}}{\|W_{l}\|_{F}}\bigg)-1}\\\leq{\prod_{l=1}^{L} \bigg(1+\frac{\|\Delta W_{l}\|_{F}}{\|W_{l}\|_{F}}\bigg)-1}
    \end{multline}
    Now using the results of the descent lemma (\ref{descentLem}), descent is guaranteed if:
    \begin{align}\label{AlmDone}
        \prod_{l=1}^{L} \bigg(1+\frac{\|\Delta W_{l}\|_{F}}{\|W_{l}\|_{F}}\bigg)< 1+cos(\theta_{k})\nonumber \\ (for\ all\ k = 1,...,L)
    \end{align}
    where $\theta_{k}$ measures the angle between $\Delta W$ and $-g_{k}(W)$. Considering (\ref{UpdL}), the perturbation is defined as $\Delta W = -\eta |W| \odot sign\ g(W)$ where $\|\Delta W_{*}\|_{F}/\|W_{*}\|_{F} = \eta$ for any possible subset of weights $W_{*}$. Let $\measuredangle(\centerdot, \centerdot)$ denote the angle between its operators. Now, $\theta_{k}$ and $\gamma_{k}$ are related by:
    \begin{align*}
       \theta_{k} &= \measuredangle(\Delta W_{k},-g_{k}(W))\\
        &= \measuredangle(-\eta\ |W_{k}| \odot sign(g_{k}(W)),-g_{k}(W))\\
        &= \measuredangle(|W_{k}| \odot sign(g_{k}(W)),g_{k}(W))\\
        &= \measuredangle(|W_{k}|,|g_{k}(W)|)\\
        &=\gamma_{k}
    \end{align*}
    Substituting the two previous results into (\ref{AlmDone}), we get:
    \begin{equation*}
        \prod_{l=1}^{L} (1+\eta) < 1+cos(\gamma_{k})
    \end{equation*}
    which becomes
    \begin{equation*}
        (1+\eta)^{L} <1+cos(\gamma_{k})
    \end{equation*}
    and finally
    \begin{equation}
        \eta < (1+cos(\gamma_{k})^{\frac{1}{L}}-1 
    \end{equation}
\end{proof}

One final remark for this section is that the update law used in the proof and that used in the algorithm are different. It is worth mentioning that for the practical algorithm we take advantage of the approximation: $w(1\pm \eta) \approx we^{\pm \eta}$.

\section{Stability Analysis}
The error dynamics were presented in section 3, culminating in (\ref{erdotfinal}). The results of section 4 support the embedded DNN's capacity to learn and approximate a function to a certain degree (proves that the actual weights approach the true weights over time) and exerts hard bounds on the weights. However the error dynamics of the function still require analysis. Consider the Lyapunov candidate function:
\begin{equation}
    L = e_{r}^{T}e_{r}
\end{equation}
Taking the derivative:
\begin{equation*}
    \dot{L}=e_{r}^{T}\dot{e}_{r} + \dot{e}_{r}^{T}e_{r}
\end{equation*}
\begin{multline}\label{ldot1}
    \dot{L}=e_{r}^{T}(\tilde{f}(X(t)) - K(X(t)-X_{d}(t))) +\\ (\tilde{f}(X(t)) - K(X(t)-X_{d}(t)))^{T}e_{r}
\end{multline}
Considering $e_{r} = X(t)-X_{d}(t)$, (\ref{ldot1}) becomes:
\begin{equation*}
    \dot{L}=e_{r}^{T}(\tilde{f}(X(t)) - Ke_{r})+ (\tilde{f}(X(t)) - Ke_{r})^{T}e_{r}
\end{equation*}
\begin{equation*}
    \dot{L}=e_{r}^{T}\tilde{f}(X(t)) - e_{r}^{T}Ke_{r} + \tilde{f}(X(t))^{T}e_{r} - e_{r}^{T}K^{T}e_{r}
\end{equation*}
\begin{equation}\label{ldot2}
    \dot{L}=2e_{r}^{T}\tilde{f}(X(t)) - 2e_{r}^{T}K^{T}e_{r}
\end{equation}
By Cauchy-Schwartz inequality:
\begin{equation}\label{csineq}
    |e_{r}^{T}\tilde{f}(X(t))|\leq{\|e_{r}\|_{2}\|\tilde{f}(X(t))\|_{2}}
\end{equation}
Because of the results of section 4 and thus the stability guarantee of the DNN for appropriate parameter choices, $\|\tilde{f}(X(t))\|_{2}\leq{\epsilon}$, which in turn makes (\ref{csineq}):
\begin{equation}\label{csineqfinal}
    |e_{r}^{T}\tilde{f}(X(t))|\leq{\|e_{r}\|_{2}\|\tilde{f}(X(t))\|_{2}}\leq{\|e_{r}\|_{2}\epsilon}
\end{equation}
Let the eigenvalues of the gain matrix $K$ be ordered such that $\lambda_{1}\leq{\lambda_{2}}\leq{}...\leq{\lambda_{n}}$. Then, since $K>0$, the following inequality holds:
\begin{equation}\label{eigineq}
    \lambda_{1}(K)e_{r}^{T}e_{r}\leq{e_{r}^{T}Ke_{r}}\leq{\lambda_{n}(K)e_{r}^{T}e_{r}}
\end{equation}
Using $e_{r}^{T}e_{r}=\|e_{r}\|_{2}^{2}$ and substituting the upper bound of (\ref{csineqfinal}) and the lower bound of (\ref{eigineq}) into (\ref{ldot2}), we arrive at:
\begin{equation}
    \dot{L}\leq{2\|e_{r}\|_{2}\epsilon-2\lambda_{1}(K)\|e_{r}\|_{2}^{2}}
\end{equation}
Which implies $\dot{L}\leq{0}$ when:
\begin{equation*}
    \|e_{r}\|_{2}\epsilon\leq{\lambda_{1}(K)\|e_{r}\|_{2}^{2}}
\end{equation*}
\begin{equation}\label{errbnd}
    \|e_{r}\|_{2}\geq{\frac{\epsilon}{\lambda_{1}(K)}}
\end{equation}
Equation (\ref{errbnd}) implies the system is uniformly ultimately bounded (UUB) with the bound on the tracking error ($e_{r}$) magnitude depending on the upper bound of the neural network approximation error ($\epsilon$) and the smallest eigenvalue of the gain matrix $K$ ($\lambda_{1}(K)$). Because the weight update law guarantees descent and thus stability in learning and because the controller formulation was based on stable error dynamics, the system is always stable if $\lambda_{1}(K)>0$.

\section{Robot Manipulator Simulation Results}
% - describe robot manipulator system, simulation parameters, desired trajectory, maybe discuss controller design for this specific application
% - show results
% - discuss algorithm performance
% - discuss pros and cons

To test the performance of the proposed controller, a simulation was performed using a 2 link planar robot manipulator. The dynamics of this system are given below as:
\begin{equation}\label{RobDy}
    M(q)\Ddot{q}+V_{m}(q,\dot{q})\dot{q} + G(q) = \tau(t)
\end{equation}
where $q \in \mathbb{R}^{2}$ are the joint angles, $M(q)$ is the inertia matrix, $V_{m}(q,\dot{q})$ is the Coriolis/centripetal matrix, $G(q)$ is the gravity vector and $\tau_{t}$ is the control torque. (\ref{RobDy}) can be expressed in matrix form via:
\begin{equation}
    \begin{bmatrix}
\dot{q}\\
\Ddot{q}
\end{bmatrix} = \begin{bmatrix}
\dot{q}\\
-M^{-1}(q)(V_{m}(q,\dot{q})\dot{q} + G(q))
\end{bmatrix}+ \begin{bmatrix}
0\\
M^{-1}(q)
\end{bmatrix}\tau_{t}
\end{equation}
It is desired to have the joint angles follow the reference trajectory of:
\begin{equation*}
    \begin{bmatrix}
    q_{1,d}\\q_{2,d}
    \end{bmatrix}
    = \begin{bmatrix}
    sin(0.5t)\\cos(0.5t)
    \end{bmatrix}
\end{equation*}
The manipulator links had length $l_{1} = 1 m, l_{2} = 1 m$ and masses $m_{1} = 1 kg, m_{2} = 2.3 kg$. The controller PD gains and MSO gains were:
\begin{equation*}
K = 
    \begin{bmatrix}
    7.5& 0& 0& 0\\0& 7.5& 0& 0\\0& 0& 7.5& 0\\0& 0& 0& 7.5
    \end{bmatrix}, 
    \quad
    K_{2} = 
    \begin{bmatrix}
   3 0& 0& 0& 0\\0& 30& 0& 0\\0& 0& 30& 0\\0& 0& 0& 30
    \end{bmatrix}
\end{equation*}
The neural network used had four layers, sigmoidal hidden layer activation functions and its weights were initialised as $W\sim{N(0,\frac{3.6}{\sqrt{L}})}$ according to \cite{DBLP:journals/corr/Kumar17}.
The weight update parameters were $\eta = 0.001$, $\eta_{max}=100\eta$, $\sigma_{max} = 1250$, $B = 0.999$. The simulation time was 10 seconds, the integration time step was 0.001. The initial conditions were:
\begin{equation*}
    \begin{bmatrix}
    q_{1,0}\\q_{2,0}\\\dot{q}_{1,0}\\\dot{q}_{2,0}
    \end{bmatrix}
    =
    \begin{bmatrix}
    0.5\\0.5\\0\\0
    \end{bmatrix},
    \quad
    \hat{q}_0 = q_{0}
\end{equation*}
Fig. \ref{NoNN}-\ref{InTor} display the results of the simulation. Fig. \ref{NoNN} shows the tracking performance of the system with the neural network turned off. While the actual trajectories loosely follow the desired trajectories, it is obvious that there is significant error in the tracking performance. In Fig. \ref{Traj}, the desired and actual state trajectories are displayed, this time using the DNN-MSO to learn the system. The performance is much better when compared to Fig. \ref{NoNN} with the actual states converging to the desired states very quickly and showing excellent tracking throughout the simulation. The trajectory tracking error can be seen in Fig. \ref{Error}. Similar to the trajectory figure, the error figure displays the fast convergence and that the error stays very close to zero throughout time with some very minor oscillations in the error 1 signal. The torque inputs, Fig. \ref{InTor}, seem to be reasonable with moderate amplitude.

\begin{figure}[h!]
\centering
\includegraphics[width=3.45in]{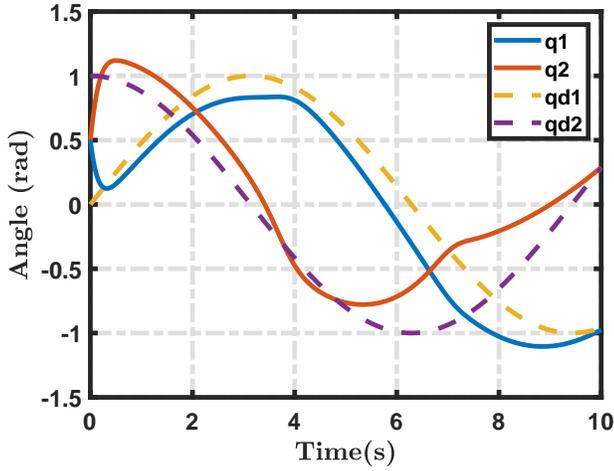}
\caption{Trajectory Tracking: No NN}
\label{NoNN}
\end{figure}

\begin{figure}[h!]
\centering
\includegraphics[width=3.45in]{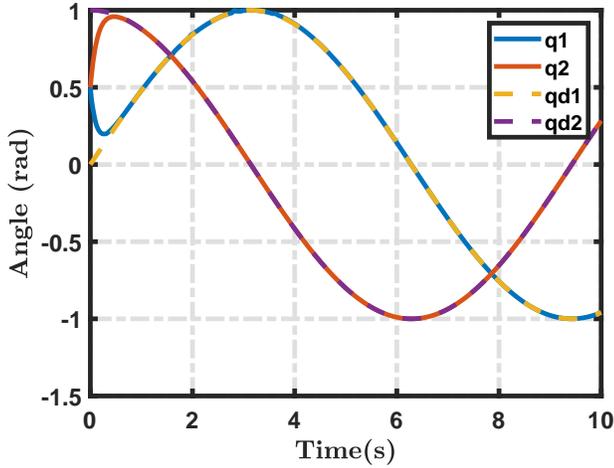}
\caption{Trajectory Tracking: with NN}
\label{Traj}
\end{figure}

\begin{figure}[h!]
\centering
\includegraphics[width=3.45in]{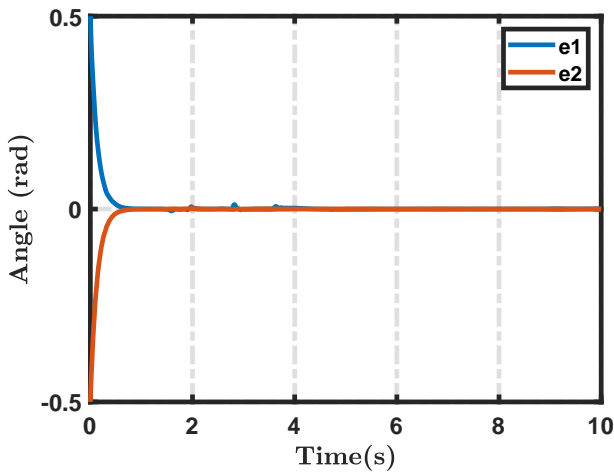}
\caption{Error: with NN}
\label{Error}
\end{figure}

\begin{figure}[h!]
\centering
\includegraphics[width=3.45in]{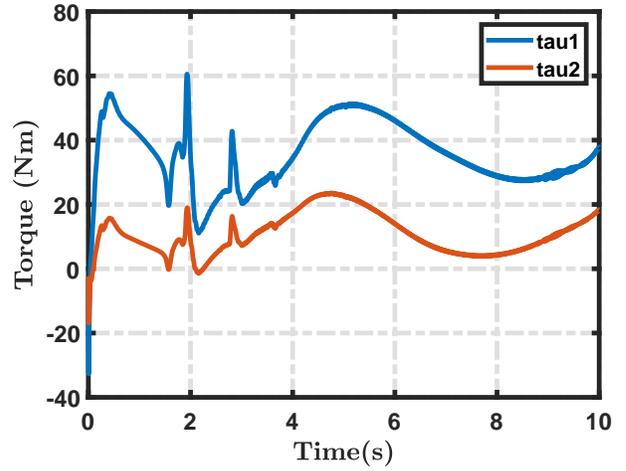}
\caption{Torque: with NN}
\label{InTor}
\end{figure}

\section{Conclusion}
% - recap the paper, any closing remarks, future work discussion
This paper presented a novel model-free dynamic inversion controller featuring a DNN that was trained online using the multiplicative weight update training algorithm. The derivation of the controller was given as well as a stability analysis for a tracking problem and the update law used to train the network was shown and discussed. A simulation showing a joint angle tracking problem with a highly nonlinear robotic manipulator demonstrated the effectiveness of the controller. The potential of a DNN in a nonlinear controller was displayed as the network was able to learn quickly and effectively, leading to exceptional tracking performance as compared to a standard non-adaptive feedback controller. To the author's knowledge, this is the first work to feature a truly online training deep neural network in a dynamic inversion controller.

Future work includes implementing this controller on a real plant and examining the results. Real world implementation comes with its own set of challenges namely overcoming delay that might be present due to hardware restrictions and thus might affect the controller performance. Another problem to overcome is the algorithm's sensitivity to initial conditions. This was somewhat overcome in this paper by utilizing the initialization strategy outlined by \cite{DBLP:journals/corr/Kumar17} but there exists much room for improvement to achieve consistent results. Finally, the author's would like to extend the theory presented in this paper to backstepping control to be more utilizable in general nonlinear systems.

\bibliographystyle{IEEEtran}
\bibliography{ref}

\begin{thebibliography}{10}
\providecommand{\url}[1]{#1}
\csname url@samestyle\endcsname
\providecommand{\newblock}{\relax}
\providecommand{\bibinfo}[2]{#2}
\providecommand{\BIBentrySTDinterwordspacing}{\spaceskip=0pt\relax}
\providecommand{\BIBentryALTinterwordstretchfactor}{4}
\providecommand{\BIBentryALTinterwordspacing}{\spaceskip=\fontdimen2\font plus
\BIBentryALTinterwordstretchfactor\fontdimen3\font minus
  \fontdimen4\font\relax}
\providecommand{\BIBforeignlanguage}[2]{{%
\expandafter\ifx\csname l@#1\endcsname\relax
\typeout{** WARNING: IEEEtran.bst: No hyphenation pattern has been}%
\typeout{** loaded for the language `#1'. Using the pattern for}%
\typeout{** the default language instead.}%
\else
\language=\csname l@#1\endcsname
\fi
#2}}
\providecommand{\BIBdecl}{\relax}
\BIBdecl

\bibitem{rajagopal_balakrishnan_nguyen_krishnakumar_mannava_2009}
K.~Rajagopal, S.~Balakrishnan, N.~Nguyen, K.~Krishnakumar, and A.~Mannava,
  ``Neuroadaptive model following controller design for non-affine and
  non-square aircraft systems,'' \emph{AIAA Guidance, Navigation, and Control
  Conference}, 2009.

\bibitem{padhi_balakrishnan_unnikrishnan_2007}
R.~Padhi, S.~Balakrishnan, and N.~Unnikrishnan, ``Model-following
  neuro-adaptive control design for non-square, non-affine nonlinear systems,''
  \emph{IET Control Theory and Applications}, vol.~1, no.~6, p. 1650–1661,
  2007.

\bibitem{ghafoor_balakrishnan_jagannathan_yucelen_2018}
A.~Ghafoor, S.~Balakrishnan, S.~Jagannathan, and T.~Yucelen, ``Event triggered
  neuro-adaptive controller (etnac) design for uncertain linear systems,''
  \emph{2018 IEEE Conference on Decision and Control (CDC)}, 2018.

\bibitem{6891229}
Z.~{Li}, Y.~{Xia}, C.~{Su}, J.~{Deng}, J.~{Fu}, and W.~{He}, ``Missile guidance
  law based on robust model predictive control using neural-network
  optimization,'' \emph{IEEE Transactions on Neural Networks and Learning
  Systems}, vol.~26, no.~8, pp. 1803--1809, 2015.

\bibitem{8559329}
Y.~{Wang}, J.~{Tong}, T.~{Song}, and Z.~{Wan}, ``Unmanned surface vehicle
  course tracking control based on neural network and deep deterministic policy
  gradient algorithm,'' in \emph{2018 OCEANS - MTS/IEEE Kobe Techno-Oceans
  (OTO)}, 2018, pp. 1--5.

\bibitem{9097881}
X.~{Wang}, X.~{Ye}, and W.~{Liu}, ``Online adaptive critic learning control of
  unknown dynamics with application to deep submergence rescue vehicle,''
  \emph{IEEE Access}, vol.~8, pp. 96\,565--96\,580, 2020.

\bibitem{8489118}
Y.~{Liu}, Y.~{Zhou}, and X.~{Li}, ``Attitude estimation of unmanned aerial
  vehicle based on lstm neural network,'' in \emph{2018 International Joint
  Conference on Neural Networks (IJCNN)}, 2018, pp. 1--6.

\bibitem{7867471}
{Huang Yi}, {Sun Shiyu}, {Duan Xiusheng}, and {Chen Zhigang}, ``A study on deep
  neural networks framework,'' in \emph{2016 IEEE Advanced Information
  Management, Communicates, Electronic and Automation Control Conference
  (IMCEC)}, 2016, pp. 1519--1522.

\bibitem{8203866}
H.~{Li}, T.~{Wei}, A.~{Ren}, Q.~{Zhu}, and Y.~{Wang}, ``Deep reinforcement
  learning: Framework, applications, and embedded implementations: Invited
  paper,'' in \emph{2017 IEEE/ACM International Conference on Computer-Aided
  Design (ICCAD)}, 2017, pp. 847--854.

\bibitem{8693516}
J.~J.~Q. {Yu}, W.~{Yu}, and J.~{Gu}, ``Online vehicle routing with neural
  combinatorial optimization and deep reinforcement learning,'' \emph{IEEE
  Transactions on Intelligent Transportation Systems}, vol.~20, no.~10, pp.
  3806--3817, 2019.

\bibitem{8682051}
C.~{Liang}, W.~{Wang}, Z.~{Liu}, C.~{Lai}, and B.~{Zhou}, ``Learning to guide:
  Guidance law based on deep meta-learning and model predictive path integral
  control,'' \emph{IEEE Access}, vol.~7, pp. 47\,353--47\,365, 2019.

\bibitem{8757756}
M.~{Xin}, Y.~{Gao}, T.~{Mou}, and J.~{Ye}, ``Online hybrid learning to speed up
  deep reinforcement learning method for commercial aircraft control,'' in
  \emph{2019 3rd International Symposium on Autonomous Systems (ISAS)}, 2019,
  pp. 305--310.

\bibitem{7487175}
T.~{Zhang}, G.~{Kahn}, S.~{Levine}, and P.~{Abbeel}, ``Learning deep control
  policies for autonomous aerial vehicles with mpc-guided policy search,'' in
  \emph{2016 IEEE International Conference on Robotics and Automation (ICRA)},
  2016, pp. 528--535.

\bibitem{8794314}
A.~{Sarabakha} and E.~{Kayacan}, ``Online deep learning for improved trajectory
  tracking of unmanned aerial vehicles using expert knowledge,'' in \emph{2019
  International Conference on Robotics and Automation (ICRA)}, 2019, pp.
  7727--7733.

\bibitem{9123670}
B.~{Karg} and S.~{Lucia}, ``Efficient representation and approximation of model
  predictive control laws via deep learning,'' \emph{IEEE Transactions on
  Cybernetics}, vol.~50, no.~9, pp. 3866--3878, 2020.

\bibitem{8809217}
A.~{Sarabakha} and E.~{Kayacan}, ``Online deep fuzzy learning for control of
  nonlinear systems using expert knowledge,'' \emph{IEEE Transactions on Fuzzy
  Systems}, vol.~28, no.~7, pp. 1492--1503, 2020.

\bibitem{8939337}
R.~{Chai}, A.~{Tsourdos}, A.~{Savvaris}, S.~{Chai}, Y.~{Xia}, and C.~L.~P.
  {Chen}, ``Six-dof spacecraft optimal trajectory planning and real-time
  attitude control: A deep neural network-based approach,'' \emph{IEEE
  Transactions on Neural Networks and Learning Systems}, vol.~31, no.~11, pp.
  5005--5013, 2020.

\bibitem{bernstein2021learning}
J.~Bernstein, J.~Zhao, M.~Meister, M.-Y. Liu, A.~Anandkumar, and Y.~Yue,
  ``Learning compositional functions via multiplicative weight updates,'' 2021.

\bibitem{bernstein2021distance}
J.~Bernstein, A.~Vahdat, Y.~Yue, and M.-Y. Liu, ``On the distance between two
  neural networks and the stability of learning,'' 2021.

\bibitem{DBLP:journals/toc/AroraHK12}
\BIBentryALTinterwordspacing
S.~Arora, E.~Hazan, and S.~Kale, ``The multiplicative weights update method: a
  meta-algorithm and applications,'' \emph{Theory of Computing}, vol.~8, no.~1,
  pp. 121--164, 2012. [Online]. Available:
  \url{https://doi.org/10.4086/toc.2012.v008a006}
\BIBentrySTDinterwordspacing

\bibitem{DBLP:journals/corr/Kumar17}
\BIBentryALTinterwordspacing
S.~K. Kumar, ``On weight initialization in deep neural networks,'' \emph{CoRR},
  vol. abs/1704.08863, 2017. [Online]. Available:
  \url{http://arxiv.org/abs/1704.08863}
\BIBentrySTDinterwordspacing

\end{thebibliography}
% \printbibliography

\end{document}